\documentclass[10pt]{article} 
\usepackage[preprint]{tmlr}
 
\usepackage{amsmath,soul}
\usepackage{adjustbox}

\usepackage[utf8]{inputenc} 
\usepackage[T1]{fontenc}    
\usepackage{hyperref}       
\usepackage{url}            
\usepackage{booktabs}       
\usepackage{amsfonts}       
\usepackage{nicefrac}       
\usepackage{microtype}      
\usepackage{xcolor}         
\usepackage{geometry}
\usepackage{color}
\usepackage{float}
\usepackage{booktabs}
\usepackage{textcomp}
\usepackage{mathtools}
\usepackage{bm}
\usepackage{amsmath}
\usepackage{amsthm}
\usepackage{amssymb}
\usepackage{graphicx}
\usepackage{booktabs,tabularx}
\usepackage{makecell,ragged2e,cleveref}
\usepackage{xspace}
\newcolumntype{Y}{>{\RaggedRight\arraybackslash}X}

\theoremstyle{plain}
\newtheorem{thm}{Theorem \protect\theoremname}
\newtheorem{theorem}[thm]{Theorem}
\newtheorem{defn}[thm]{Definition}
\newtheorem{prop}[thm]{Proposition}

\newtheorem{lem}[thm]{Lemma}
\newtheorem{cor}[thm]{Corollary}

\newcommand{\ylnote}[1]{{\color{orange}[ Yang: #1 ]}}

\newcommand{\WGP}{\textsc{WGP}\xspace}
\newcommand{\SWGP}{\textsc{SWGP}\xspace}
\newcommand{\PWA}{\textsc{PWA}\xspace}
\newcommand{\PWAGP}{\textsc{PWA-GP}\xspace}
\newcommand{\PWAGPs}{\textsc{PWA-GPs}\xspace}
\newcommand{\PCPWA}{\textsc{PCPWA}\xspace}
\newcommand{\UIGP}{\textsc{UI-GP}\xspace}

\PassOptionsToPackage{authoryear}{natbib}
\title{Wasserstein-type Gaussian Process Regressions for Input
Measurement Uncertainty}

\author{%
  Hengrui Luo\\
  Department of Statistics\\
  Rice University, Houston, TX 77005 \\
  \texttt{hl180@rice.edu}\\ 
  Xiaoye S. Li, Yang Liu, Marcus Noack, Ji Qiang, Mark D. Risser\\
  \texttt{xsli@lbl.gov}, \texttt{liuyangzhuan@lbl.gov}, \texttt{marcusnoack@lbl.gov}, 
  \texttt{jqiang@lbl.gov}, \texttt{mdrisser@lbl.gov}\\
  Lawrence Berkeley National Laboratory,  Berkeley, CA 94709 \\
}

\begin{document}

\maketitle

\begin{abstract}
Gaussian process (GP) regression is widely used for uncertainty quantification, yet the standard formulation assumes noise-free covariates. When inputs are measured with error, this errors-in-variables (EIV) setting can lead to optimistically narrow posterior intervals and biased decisions. We study GP regression under input measurement uncertainty by representing each noisy input as a probability measure and defining covariance through Wasserstein distances between these measures. Building on this perspective, we instantiate a deterministic projected Wasserstein ARD (PWA) kernel whose one-dimensional components admit closed-form expressions and whose product structure yields a scalable, positive-definite kernel on distributions. Unlike latent-input GP models, PWA-based GPs (\PWAGPs) handle input noise without introducing unobserved covariates or Monte Carlo projections, making uncertainty quantification more transparent and robust. 
\end{abstract}

\section{Introduction}

\label{sec:intro}

Gaussian Processes (GPs) are core methods for stochastic function approximation in machine learning and surrogate modeling. They are fully defined by a prior mean and a covariance kernel, making them analytically tractable and able to encode function smoothness \citep{williams2006gpml}. GPs output predictive distributions rather than point estimates and quantify uncertainty and confidence. In a Bayesian framework GPs also propagate hyperparameter uncertainty. They have been applied in forward modeling \citep{deisenroth2010efficient,vinogradska2016stability}, uncertainty quantification \citep{tuo2022uncertainty,wang2021inference}, 
parameter optimization~\citep{gptune:ppopp21,ipac2024-thpc72},
and autonomous experimentation \citep{noack2019kriging,stach2021autonomous,noack2021gaussian,thomas2022autonomous}.
Standard GP regression assumes noise-free inputs $X$, modeling observations from $f:\mathbb R^d\rightarrow\mathbb R$ as
\[
y = f(X) + \epsilon,\quad \epsilon\sim\mathcal N(0,\sigma^2).
\]
When the inputs are noisy, i.e., we observed $U = X + \epsilon_X$ as inputs, this “errors‐in‐variables” (EIV) setting can degrade performance \citep{zhou2019gaussian}. Figure~\ref{fig:Illustration-of-error-in-variable} shows how input uncertainty can lead a GP failing to estimate $f$. Thus, to recover the true function for noisy input, GPs must incorporate input uncertainty or address EIV.
The following proposition describes the effect of EIV in terms of the uncertainty quantification of prediction. 


\begin{table*}[t]
\centering
\tiny
\setlength{\tabcolsep}{3pt}
\renewcommand{\arraystretch}{1.18}
\begin{tabularx}{\textwidth}{@{}p{2.55cm} p{2.35cm} c c c Y p{2.75cm}@{}}
\toprule
\textbf{Method (paper name)} &
\textbf{Kernel input object} &
\textbf{Dist.} &
\textbf{Latent $X$} &
\textbf{Det.} &
\textbf{EIV mechanism } &
\textbf{Citations} \\
\midrule

\textbf{Regular GP} \newline (RBF / Mat\'ern / Exp) &
Point estimate (empirical mean $m_i$) &
N & N & Y &
Ignores input uncertainty (variance / shape); standard baseline in Sec.~3 and Table~1. &
\cite{williams2006gpml} \\

\textbf{Aggregated Regular GPs} &
Noisy samples $\{U_{ij}\}$ as point inputs \newline (fit per-$j$ GP; average) &
N & N & Y &
Ensemble over replicate-specific GPs; can amplify variance in sparse regions (as observed). &
\textbf{This paper.}  \\

\textbf{IV-function / cov-only encoding}  &
Gaussian summary $(\mu_x,\Sigma_x)$ \newline but only $\Sigma_x$ used in kernel &
Y* & N & Y &
Efficient for locally Gaussian perturbations; cannot represent skew / multimodality / complex support. &
\cite{morenomunoz2018heterogeneous} \\

\midrule

\textbf{WGP} \newline (Eq.~(5), WGP-RBF) &
Measure $\mu$   &
Y & N & Y &
OT-based kernel on distributions; PD only in special cases (e.g., Gaussian $p{=}2$, or 1D). &
\cite{candelieri2022gaussian,peyre2019computational} \\

\textbf{sliced WGP} \newline (Eq.~(6), SWGP) &
Measure $\mu$; average $W_p$ over random 1D projections $u$ &
Y & N & \textbf{MC} &
PD preserved via expectation over projections; stochastic projections add Monte Carlo variability and cost. &
\cite{meunier2022slicedwasserstein,bonet2023sliced} \\

\textbf{PWA-GP} \newline (Eq.~(7)) &
Coordinate-wise 1D marginals $\{\mu_i\}_{i=1}^d$ &
Y & N & Y &
Product of 1D Wasserstein kernels with dimension-specific hyperparameters (ARD-like); PD via product closure. &
\textbf{This paper.} \\

\textbf{PCPWA} \newline (Eq.~(8)) &
Fixed directions $\{v_r\}_{r=1}^m$ \newline (e.g., PCs); 1D pushforwards $\mu^{v_r}$ &
Y & N & Y &
Same idea as PWA but along a fixed orthonormal basis (often PCA); trades expressivity vs. cost via $m$. &
\textbf{This paper.} \\

\midrule

\textbf{Latent-input calibration GP} &
Noisy $U$ with latent true $X$ (integrate out) &
N & Y & (inf.) &
Classic latent-variable EIV route; requires integration / sampling / approximations. &
\cite{kennedy2001calibration} \\

\textbf{Uncertain-input GP} &
Input distribution (often Gaussian) used in prediction / propagation &
Y & (approx) & Y &
Uses uncertain regressors (often Gaussian) in forecasting; typically relies on approximations/closed forms. &
\cite{girard2002gaussian} \\

\textbf{GP training w/ input noise} &
Noisy inputs + local approximation during training &
N & (approx) & Y &
Convolves/approximates input noise effect during training (e.g., local linearization). &
\cite{mchutchon2011gaussian} \\

\textbf{Deep GP (latent-variable)} &
Joint latent representation for inputs/outputs &
N & Y & (inf.) &
Highly flexible latent EIV modeling; inference heavier / approximate. &
\cite{damianou2013deep} \\

\midrule

\textbf{Low-rank Fr\'echet EIV regression} &
Noisy covariates w/ low-rank approximation &
N & N & Y &
Non-GP competitor class mentioned (not compared in experiments); linear/low-rank structure. &
\cite{song2023frechet} \\

\textbf{KME / MMD distribution regression}  &
Empirical samples per input distribution; RKHS mean embeddings &
Y & N & Y &
Strong non-OT distribution-input baseline; geometry differs from Wasserstein (often easier compute). &
\cite{szabo2016distreg,muandet2017kme} \\

\bottomrule
\end{tabularx}
\caption{Compact comparison of EIV / uncertain-input methods mentioned in the manuscript (including experimental baselines and Wasserstein variants) plus an added distribution-input competitor family (KME/MMD). Dist. indicates whether the model explicitly consumes a distribution/set input (Y* = Gaussian-summary distribution but only covariance used). Det.: deterministic vs Monte Carlo (MC) due to random projections. OT = optimal transport;
PD = positive definite;
KME = kernel mean embedding;
MMD = maximum mean discrepancy;
PCA = principal component analysis.}
\label{tab:eiv_methods_compact}
\end{table*}

The next result provides the technical link between the proposed Wasserstein GP and uncertainty quantification, it shows that a standard GP can undercover when input noise is ignored. 
\begin{prop}\label{ref:descriptive_prop}
Consider the errors-in-variables model
\[
Y = f(X+\varepsilon_X) + \varepsilon,\qquad
\varepsilon_X \sim \mathcal{N}(0,\Sigma_X),\quad
\varepsilon \sim \mathcal{N}(0,\sigma^2),
\]
with $\varepsilon_X \perp \varepsilon$. Let $f(x)=c+w^\top x$ be affine with $w\neq 0$.
Define the naive $(1-\alpha)$ interval that ignores input noise:
\[
I_\alpha := [\,Y - z_{1-\alpha/2}\sigma,\; Y + z_{1-\alpha/2}\sigma\,].
\]
Then for any fixed $X$,
\[
\mathbb{P}\bigl(f(X)\in I_\alpha \mid X\bigr)
=
2\Phi\!\left(\frac{z_{1-\alpha/2}\sigma}{\sqrt{\sigma^2 + w^\top\Sigma_X w}}\right)-1
\;<\; 1-\alpha
\]
whenever $w^\top\Sigma_X w>0$.
\end{prop}

Proposition~\ref{ref:descriptive_prop} formalizes a simple failure mode of a
naive GP in an errors-in-variables setting. The interval $I_\alpha$ accounts
only for the output noise variance $\sigma^2$, while the true target
$f(X)=c+w^\top X$ also varies through the input error term
$w^\top \varepsilon_X$, whose variance is $w^\top \Sigma_X w$. As a result,
the nominal $(1-\alpha)$ interval is too narrow and its actual coverage is
strictly smaller than $1-\alpha$. In this paper, we refer to this phenomenon
as \emph{undercoverage}.

One approach to taking into account EIV uses latent‐variable models: \citet{kennedy2001calibration} treat true inputs as latent and integrate them out; \citet{girard2002gaussian} apply this to time series; \citet{damianou2013deep} model inputs and outputs jointly; and \citet{mchutchon2011gaussian} convolve the GP prior with an input‐noise distribution. These methods capture input uncertainty but require high‐dimensional inference, which can be slow or unstable in non‐Gaussian settings, and often rely on numerical approximations that may bias results.

Alternatively, one can use kernels on probability measures. If each input is treated as a distribution, for example $U\sim\mathcal N(\mu_X,\Sigma_X)$, kernels based on Wasserstein distances compare these inputs directly \citep[see][]{Panaretos2019}. \citet{candelieri2022gaussian} study univariate discrete measures, and \citet{meunier2022slicedwasserstein,bonet2023sliced} develop sliced Wasserstein kernels for general distributions.

In this work, we view each noisy input $U_i$ as a probability measure $\mu_i$ and place a GP prior directly on the space of such measures. We instantiate this idea with the \PWA kernel, which preserves the familiar automatic relevance determination (ARD) structure while encoding input uncertainty through Wasserstein distances between one-dimensional marginals. Intuitively, when the input noise is small, the \PWAGP behaves similarly to a GP fitted on the (latent) noise-free inputs $(X,Y)$, whereas increasing dispersion of the input distributions naturally inflates posterior uncertainty, leading to more honest predictive intervals in error-in-variables problems.

\begin{figure}[t!]
\centering \includegraphics[width=0.40\textwidth]{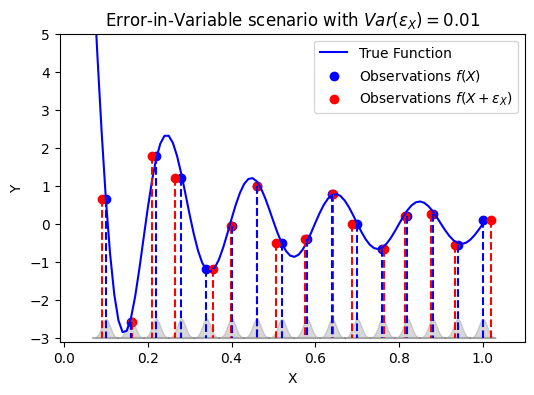}\includegraphics[width=0.40\textwidth]{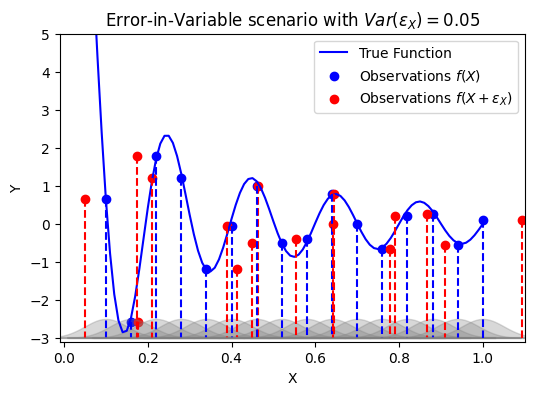}
\caption{\label{fig:Illustration-of-error-in-variable}Illustration of error-in-variable
regression problem. In both panels, the true function is $y=f(X)=\frac{\sin(10\pi\cdot X)}{2X}+(X-1)^{4}$,
but in each case the \textquotedblleft true\textquotedblright{} input locations 
$X$ are contaminated with measurement errors (with standard deviation
of 0.01 and 0.05 on the left and right, respectively). Using a GP to accurately infer the true function (the blue line) must
account for the fact that the input locations $U$ are uncertain.} 
\end{figure}
In the \PWA kernel \eqref{eq:ker_GWp-generic}, separate kernel parameters for each input dimension act as ARD-style length scales, handling dimension-specific uncertainty via an automatic relevance determination construction~\citep{williams2006gpml}. This avoids full multivariate optimal transport and randomized slicing while still capturing anisotropic input variability in a way that directly serves uncertainty quantification rather than merely introducing a new kernel form.

In EIV problems, finite-slice implementations of sliced Wasserstein kernels \citep{meunier2022slicedwasserstein} introduce additional Monte Carlo approximation error on top of measurement noise. In contrast, our \PWA construction uses deterministic (coordinate or fixed-basis) projections, avoiding slice-induced randomness while retaining valid projected kernels and ARD flexibility. Our theoretical analysis focuses on the case $p=1$, where covering-number arguments yield uniform error bounds. We summarize related methods in Table~\ref{tab:eiv_methods_compact}, discuss computational complexities in Appendix~\ref{sec:complexity}, and compare a representative subset in our uncertainty-quantification experiments.

\textbf{Summary of our contributions.}

\begin{itemize}
    \item We focus on uncertainty quantification for GP regression with input measurement error, modeling each noisy input as a probability measure and working in the Wasserstein space of distributions. In a simple linear-Gaussian EIV model we show that standard GPs that ignore input noise produce predictive intervals that strictly undercover.
    
    \item We develop a deterministic \PWA kernel and its principal-component variant \PCPWA, which define GPs directly on distributions via products of closed-form one-dimensional Wasserstein kernels. The construction preserves automatic relevance determination (ARD) while avoiding latent variables and stochastic slicing schemes commonly used for error-in-variables GPs.
    
    \item We establish a deterministic net-extension bound for the posterior mean of PWA-based GPs and, under a well-specified GP prior, obtain a uniform high-probability posterior band controlled by the posterior standard deviation. This yields a simple sufficient condition under which the resulting credible intervals are conservative.
    
    \item Through experiments on synthetic data, accelerator calibration, and NOAA drifter trajectories, we demonstrate that \PWAGPs provide competitive point predictions and markedly improved uncertainty quantification in terms of coverage and CRPS and relative to standard kernels and sliced Wasserstein baselines under input noise.
\end{itemize}

\section{GPs with Wasserstein-type Kernels}

\label{sec:kernel_design_WGP}
 
\subsection{Wasserstein and Gromov-Wasserstein distances}
\label{sec:Wdist}

Wasserstein distances improve how distributional inputs are handled in machine learning. \citet{arjovsky2017wasserstein} highlighted their use for output distributions, but they also apply to input uncertainty.

Intuitively, imagine moving “sand” from a pile at $x$ to fill a hole at $y$: the cost is the mass moved times the distance \citep{mallasto2017learning}. We write $\mathcal P_p(X)$ for the set of Borel probability measures on $X$ with finite $p$th moment, i.e. $\int d_X(x,x_0)^p\,d\mu(x)<\infty$ for some (equivalently any) $x_0\in X$.
 Formally:

\begin{defn}
($p$-Wasserstein distance) Let $p\in[1,\infty)$ and let $\mu,\nu\in\mathcal P_p(X)$ be probability measures on a metric space $(X,d_X)$ with finite $p$th moments. The $p$-Wasserstein distance is
\begin{equation}
W_p(\mu,\nu)
=\Bigl(\inf_{\gamma\in\Gamma(\mu,\nu)}\int_{X\times X} d_X(x,y)^p\,d\gamma(x,y)\Bigr)^{1/p},
\label{eq:def_Wp}
\end{equation}
where $\Gamma(\mu,\nu)$ denotes the set of couplings with marginals $\mu$ and $\nu$. The case $p=1$ is the Earth Mover’s Distance \citep{peyre2019computational}.
\end{defn}

$W_p(\mu,\nu)$ is the minimal “effort” to transport $\mu$ to $\nu$, raising cost to the $p$th power. These distances scale polynomially in asymptotic regimes \citep{panaretos2020invitation,weed2019sharp}.

\begin{defn}
(Gromov-Wasserstein distance)
Let $\mu\in\mathcal P_p(X)$ and $\nu\in\mathcal P_p(Y)$, where $(X,d_X)$ and $(Y,d_Y)$ are metric spaces.
The $p$-Gromov-Wasserstein distance is
\begin{equation}
GW_p(\mu,\nu)
:=
\left(
\inf_{\pi\in\Gamma(\mu,\nu)}
\iint_{(X\times Y)^2}
\bigl|d_X(x,x')-d_Y(y,y')\bigr|^p\,
d\pi(x,y)\,d\pi(x',y')
\right)^{1/p}.
\label{eq:def_GW_p}
\end{equation}
\end{defn}

$GW_p$ matches the relational geometry of $X$ and $Y$, finding correspondences that minimize distortion.

In essence, $GW_p$ generalizes $W_p$ to compare distributions on different spaces by focusing on pairwise distances. $W_p$ uses a common ground metric, while $GW_p$ handles cross-domain structures, making both valuable for pratical tasks in image processing, machine learning, and statistical physics. However, we lack knowledge of the PD property for GW-type kernels.

While the statement of Wasserstein and Gromov-Wasserstein kernels
is straightforward, calculating $W_{p}$ and $GW_{p}$ can be quite
challenging. For $\bm{x}\in\mathbb{R}^{1}$, we can have
a closed-form expression for any $p\ge 1$ 
: 
\begin{prop}
\label{prop:one-dim} \citep{meunier2022slicedwasserstein} Let $\mu,\nu$ be probability measures defined
on the real line $\mathbb{R}^1$ and denote their cumulative distribution
functions by $F_{\mu}(x)$ and $F_{\nu}(x)$. Then,
\begin{align}
W_{p}(\mu,\nu)=\left(\int_{0}^{1}\left|F_{\mu}^{-1}(q)-F_{\nu}^{-1}(q)\right|^{p}\,dq\right)^{1/p}.
\end{align}
\end{prop}

\subsection{Wasserstein‐type covariance kernels}
\label{sec:sparseKernels}
We can define $k_{W_p}$ as follows, 
\begin{equation}
k_{W_p}(\mu,\nu)
=\lambda\exp\bigl(-\sigma\,W_p(\mu,\nu)^p\bigr)
\label{eq:kernel_Wp}
\end{equation}
To use $k_{W_p}$ as a GP covariance function, it is not enough that $W_p$ be a
metric. A standard sufficient condition is that the exponent
$\psi(\mu,\nu)$ be \emph{conditionally negative definite} (CND): a symmetric
function $\psi$ on $\mathcal P$ is CND if for every $n\ge 1$, every
$\mu_1,\dots,\mu_n\in\mathcal P$, and every coefficients
$c_1,\dots,c_n\in\mathbb R$ satisfying $\sum_{i=1}^n c_i=0$, one has
\[
\sum_{i,j=1}^n c_i c_j \psi(\mu_i,\mu_j)\le 0.
\]
By Schoenberg's theorem, if $\psi$ is CND and $\psi(\mu,\mu)=0$, then
$k(\mu,\nu)=\exp(-\sigma\psi(\mu,\nu))$ is positive definite for every
$\sigma>0$.

In \eqref{eq:kernel_Wp}, the exponent is $\psi(\mu,\nu)=W_p(\mu,\nu)^p$,
not $W_p(\mu,\nu)^2$ unless $p=2$. In general, $W_p(\mu,\nu)^p$ need not be
CND on $\mathbb R^d$ for arbitrary $d>1$, so positive definiteness is not
automatic outside special settings. This is why we either restrict attention
to cases where positive definiteness is known (e.g., certain one-dimensional
or Gaussian settings) or use projected/sliced constructions whose
one-dimensional factors are positive definite.


A special case is when the inputs are Gaussian and $p=2$: here the projection
is not needed, the kernel \eqref{eq:kernel_Wp} is positive definite in the
Gaussian setting, and closed forms for $W_2$ enable efficient computation.
For example,
\[
W_2^2\!\bigl(\mathcal N(\mu_x,\Sigma_x),\mathcal N(\mu_y,\Sigma_y)\bigr)
=
\|\mu_x-\mu_y\|^2
+
\operatorname{tr}\!\Bigl(
\Sigma_x+\Sigma_y
-2(\Sigma_x^{1/2}\Sigma_y\Sigma_x^{1/2})^{1/2}
\Bigr).
\]
If $\Sigma_x=\Sigma_y$, this reduces to $\|\mu_x-\mu_y\|^2$, so
\eqref{eq:kernel_Wp} becomes the usual squared-exponential
(Gaussian/RBF) kernel on the means  \citep{dowson1982frechet,peyre2019computational,horn2012matrix}.
Alternatively, one may use a \emph{sliced} kernel by averaging over 1D projections $u\in\mathbb S^{d-1}$,
\begin{equation}
k(\mu,\nu)
=\lambda\exp\!\bigl(-\sigma\,\mathbb{E}_u[W_p(\mu_u,\nu_u)^p]\bigr),
\label{eq:slice_kernel}
\end{equation}
which also remains PD \citep{meunier2022slicedwasserstein}.  Computationally we have also derived bounds on derivatives for numerical stability (See Appendix \ref{appdx:derivatives-bounds}).

\paragraph{Model names used in the paper.} We refer to the GP prior with covariance $k_{W_p}$ in \eqref{eq:kernel_Wp} as a \emph{Wasserstein GP} (\WGP). The GP using the sliced kernel \eqref{eq:slice_kernel} is denoted \SWGP. The GP using our projected ARD kernel \eqref{eq:ker_GWp-generic} is denoted \PWAGP, and its fixed-basis/principal-component variant \eqref{eq:ker_PCPWA} is denoted \PCPWA. After these first definitions we use only the abbreviations.

To ensure PD in higher dimensions without solving a $d$-dimensional optimal transportation problem, we use \emph{Projected Wasserstein ARD} kernels.  For measures $\mu,\nu$ on $\mathbb{R}^d$, 
define
\begin{equation}
k(\mu,\nu)
=\lambda\prod_{i=1}^d 
\exp\!\bigl(-\sigma_i\,W_p(\mu_i,\nu_i)^p\bigr),
\label{eq:ker_GWp-generic}
\end{equation}

where $\mu_i$ 
is the marginal distribution of measure $\mu$ along the $i$-th dimension defined on $\mathbb{R}^1$; whenever the 1D factors are PD, so is the product kernel \citep{wilson2015kernel,bilionis2013multi}. Similar construction can be done using Gromov-Wasserstein distances, yet the projections can be defined through marginals of different dimensions. Unlike previous EIV methods that introduce latent true inputs and require high‐dimensional integration or approximate marginalization \citep{kennedy2001calibration,girard2002gaussian,damianou2013deep,mchutchon2011gaussian}, PWA avoids explicit latent‐variable inference.  

Although this is a hybrid of ARD kernel and projection technique, we found this a novel construction for EIV problem that does not exist in the current literature. This is different from the sliced Wasserstein kernels \citep{meunier2022slicedwasserstein} (i.e., ``project-Wasserstein/average''), since each $(\lambda_i,\sigma_i)$ can adapt to dimension-specific uncertainty (i.e., ARD type construction), avoiding the cost of a full multivariate transport and iterative projections during the random slicing. This is also different from earlier work by  \citet{bachoc2017gaussian}, where they proposed to consider $p=2$ only; and \PWAGP is also different from  their ``PCA'' GP, which uses truncated principal component projections ARD kernel without computing Wasserstein distances over marginals (i.e., ``project-ARD''). Our proposed kernel \eqref{eq:ker_GWp-generic} is a generalization introducing the scheme of  ``Projected Wasserstein ARD''.

Beyond projecting onto coordinate axes, one may also project onto a fixed
orthonormal basis of $\mathbb R^d$, such as principal components.
Let $\{v_r\}_{r=1}^m\subset\mathbb R^d$ be fixed orthonormal directions and
define the linear functional
\[
\ell_r:\mathbb R^d\to\mathbb R,
\qquad
\ell_r(x)=v_r^\top x.
\]
For a Borel probability measure $\mu$ on $\mathbb R^d$, define its projected
marginal along $v_r$ by
\[
\mu_{v_r}:=(\ell_r)_\#\mu.
\]
Equivalently, for every Borel set $A\subset\mathbb R$,
$
\mu_{v_r}(A)
=
\mu\bigl(\{x\in\mathbb R^d: v_r^\top x\in A\}\bigr).
$
Here $(\ell_r)_\#\mu$ denotes the push-forward of $\mu$ through $\ell_r$.
For a Borel probability measure $\mu$ on $\mathbb{R}^d$, we denote by $\mu_{v_r}:=(\ell_r)_\#\mu$
the push-forward of $\mu$ through $\ell_r$, i.e., the one-dimensional marginal
of $\mu$ along direction $v_r$. Explicitly, for every Borel set
$A \subset \mathbb{R}$,
\[
\mu_{v_r}(A)
=
((\ell_r)_\#\mu)(A)
:=
\mu\bigl(\ell_r^{-1}(A)\bigr)
=
\mu\bigl(\{x \in \mathbb{R}^d : v_r^\top x \in A\}\bigr).
\]
We use the notation $(\ell_r)_\#\mu$ rather than $\ell_r(\mu)$ because
$\ell_r$ acts on points $x \in \mathbb{R}^d$, while the induced action on
measures is the push-forward map.

We define the \emph{principal-component projected Wasserstein ARD} (\PCPWA) kernel by
\begin{equation}
k_{\mathrm{PCPWA}}(\mu,\nu)
=\lambda\prod_{r=1}^m
\exp\!\bigl(-\sigma_r\,W_p(\mu_{v_r},\nu_{v_r})^p\bigr),
\label{eq:ker_PCPWA}
\end{equation}
with direction-specific length scales $(\sigma_r)_{r=1}^m$ and a global amplitude $\lambda$.

When $m=d$ and $v_r=e_r$ are the canonical basis vectors, \eqref{eq:ker_PCPWA} reduces exactly to the coordinate PWA kernel \eqref{eq:ker_GWp-generic}, so PWA is a special case of \PCPWA.  Whenever the 1D factors in \eqref{eq:ker_PCPWA} are PD, their pullbacks by $\ell_r$ are PD on $\mathcal{P}_p(\mathbb{R}^d)$, and the product over $r$ remains PD by closure under products.  Thus, \PCPWA inherits the same kernel-validity guarantees as PWA in those validated 1D settings while allowing projections onto any fixed system of directions, including principal components computed once from the design.

\subsection{\label{sec:Uniform-Error-Bound}Uniform Error Bound for Wasserstein-type GP (p=1)}
To assess models under input uncertainty, we derive a uniform bound on the posterior mean over a compact class of 1D distributions. Since our PWA kernel \eqref{eq:ker_GWp-generic} is built from 1D marginals, we limit the theory below to the case $p=1$. The deterministic net-extension argument parallels Theorem~3.1 of \citet{lederer2019uniform}, while the posterior-probability corollary stated below additionally assumes a well-specified GP prior.

Let $[a,b]\subset\mathbb R$ and let $\mathcal P$ be a set of Borel probability
measures supported on $[a,b]$ such that the Wasserstein covering numbers satisfy
$M(\tau,\mathcal P)\le C\tau^{-\alpha_0}$ for some $\alpha_0>0$.
Assume the quantile functions are uniformly Lipschitz in their argument, i.e.,
there exists $\ell>0$ such that
\[
|F_\mu^{-1}(q)-F_\mu^{-1}(q')|
\le \ell |q-q'|
\qquad
\text{for all } q,q'\in[0,1],\ \mu\in\mathcal P.
\]
Let $f:\mathcal P\to\mathbb R$ be $L_f$-Lipschitz with respect to $W_1$.
Given observations $(\mu_i,y_i)_{i=1}^N$, define
\[
K_{ij}:=k(\mu_i,\mu_j),
\qquad
k_\mu := \bigl(k(\mu,\mu_1),\dots,k(\mu,\mu_N)\bigr)^\top,
\qquad
\alpha := (K+\sigma_*^2 I_N)^{-1} y.
\]
Then
\[
\nu_N(\mu)=k_\mu^\top \alpha,
\qquad
\sigma_N^2(\mu)=k(\mu,\mu)-k_\mu^\top (K+\sigma_*^2 I_N)^{-1}k_\mu.
\]
Define
\[
L_k
:=
\sup_{\nu\in\mathcal P}\sup_{\mu\neq\mu'}
\frac{|k(\mu,\nu)-k(\mu',\nu)|}{W_1(\mu,\mu')},
\qquad
L_{\nu_N}:=N L_k \|\alpha\|_\infty,
\]
and let $\omega_{\sigma_N}$ denote a modulus of continuity for $\sigma_N$, i.e.
\[
|\sigma_N(\mu)-\sigma_N(\nu)|\le \omega_{\sigma_N}(\tau)
\qquad
\text{whenever } W_1(\mu,\nu)\le \tau.
\]
The next theorem is the deterministic extension step that turns control on a finite $\tau$-net into a uniform bound over $\mathcal P$.
\begin{theorem}\label{thm:unif_bound}
Let $\{\bar\mu_1,\ldots,\bar\mu_M\}$ be a $\tau$-net of $\mathcal P$ in $W_1$. Assume $f:\mathcal P\to\mathbb R$ is $L_f$-Lipschitz with respect to $W_1$, the kernel $k$ is $L_k$-Lipschitz in its first argument with respect to $W_1$, and $\sigma_N$ admits the modulus of continuity $\omega_{\sigma_N}$. If, for some constant $B\ge 0$,
\[
|f(\bar\mu_j)-\nu_N(\bar\mu_j)| \le B\,\sigma_N(\bar\mu_j),
\qquad j=1,\ldots,M,
\]
then for every $\mu\in\mathcal P$,
\[
|f(\mu)-\nu_N(\mu)|
\le B\,\sigma_N(\mu) + (L_f+L_{\nu_N})\tau + B\,\omega_{\sigma_N}(\tau).
\]
\end{theorem}
The proof appears in Appendix \ref{sec:proof_main_thm}, which mirrors \citet{lederer2019uniform} but imposes conditions on $L_\infty$ bounds instead of requiring its continuity, which is different from \citet{lederer2019uniform}, and one can use the 1D CDF Lipschitz nets from Proposition \ref{prop:one-dim} to bound $M(\tau,\mathcal P)\,$. 

The proof appears in Appendix \ref{sec:proof_main_thm}. It is a deterministic net-extension argument: the only stochastic ingredient enters later, when the finite-net inequalities are obtained from Gaussian posterior tails under a well-specified GP model.

\begin{cor}\label{cor:unif_bound_gp}
Assume in addition that the regression model is well specified: $f\sim\mathcal{GP}(0,k)$, the design points
$(\mu_i)_{i=1}^N$ are fixed, and the observations satisfy
\[
y_i=f(\mu_i)+\epsilon_i,\qquad \epsilon_i \stackrel{\mathrm{ind}}{\sim}\mathcal N(0,\sigma_*^2),
\qquad i=1,\dots,N.
\]
Then for any $\tau>0$ and $\delta\in(0,1)$, conditional on $\mathcal D_N$, with posterior probability at least $1-\delta$,
\[
|f(\mu)-\nu_N(\mu)|\le \sqrt{\beta(\tau)}\,\sigma_N(\mu)+\gamma(\tau)
\qquad\text{for all }\mu\in\mathcal P,
\]
where
\[
\beta(\tau)
=
\left[
\Phi^{-1}\!\left(1-\frac{\delta}{2\,M(\tau,\mathcal P)}\right)
\right]^2,
\qquad
\gamma(\tau)
=
(L_f+L_{\nu_N})\,\tau+\sqrt{\beta(\tau)}\,\omega_{\sigma_N}(\tau).
\]
\end{cor}
The same argument applies to the \PCPWA kernel with $D$ replacing $W_1$.

\begin{cor}\label{cor:unif_bound_pcpwa}
Let $\mathcal P$ be a set of Borel probability measures on $\mathbb R^d$, and define
\[
D(\mu,\nu):=\Biggl(\sum_{r=1}^m a_r\,W_1(\mu_{v_r},\nu_{v_r})^2\Biggr)^{1/2}
\]
for fixed directions $\{v_r\}_{r=1}^m\subset\mathbb R^d$ and weights $a_r>0$. Assume $f:\mathcal P\to\mathbb R$ is $L_f$-Lipschitz with respect to $D$, the \PCPWA kernel $k_{\mathrm{PCPWA}}$ from \eqref{eq:ker_PCPWA} is $L_k$-Lipschitz in its first argument with respect to $D$, and $\sigma_N$ admits the modulus of continuity $\omega_{\sigma_N}$ with respect to $D$. If $(\mathcal P,D)$ admits a $\tau$-net of size $M_{\mathrm{PCPWA}}(\tau,\mathcal P)$ and the model is well specified with prior $f\sim\mathcal{GP}(0,k_{\mathrm{PCPWA}})$, then for any $\tau>0$ and $\delta\in(0,1)$, conditional on $\mathcal D_N$, with posterior probability at least $1-\delta$,
\[
|f(\mu)-\nu_N(\mu)|
\le
\sqrt{\beta_{\mathrm{PCPWA}}(\tau)}\,\sigma_N(\mu)
+\gamma_{\mathrm{PCPWA}}(\tau)
\qquad\text{for all }\mu\in\mathcal P,
\]
where
\[
\beta_{\mathrm{PCPWA}}(\tau)
=
\left[
\Phi^{-1}\!\left(1-\frac{\delta}{2\,M_{\mathrm{PCPWA}}(\tau,\mathcal P)}\right)
\right]^2,
\qquad
\gamma_{\mathrm{PCPWA}}(\tau)=(L_f+L_{\nu_N})\tau+\sqrt{\beta_{\mathrm{PCPWA}}(\tau)}\,\omega_{\sigma_N}(\tau).
\]
\end{cor}

For $p>1$, it becomes more complicated since different projected marginals can have different covering numbers in the kernel construction \eqref{eq:ker_GWp-generic}. We therefore restrict the formal theory here to the $p=1$ setting and leave sharper high-dimensional covering estimates for future work.

The preceding corollaries are net-based posterior band statements analogous in spirit to the Wasserstein-space analysis of \citet{meunier2022slicedwasserstein}, but here the finite-net step is made explicit and the deterministic extension from the net to the full input class is separated from the model-based Gaussian tail argument. This immediately implies the following proposition, which shows how the band width can prevent the undercoverage highlighted in Proposition~\ref{ref:descriptive_prop}.

\begin{prop}\label{prop:wgp_bound}
Assume the posterior event in Corollary~\ref{cor:unif_bound_gp} holds with parameters $(\tau,\delta)$, i.e., conditional on $\mathcal D_N$,
\[
|f(\mu)-\nu_N(\mu)| \le \sqrt{\beta(\tau)}\,\sigma_N(\mu) + \gamma(\tau)
\quad\text{for all }\mu\in\mathcal P
\]
with posterior probability at least $1-\delta$. Fix $\alpha\in(0,1)$ and $z:=z_{1-\alpha/2}$. If
\[
z > \sqrt{\beta(\tau)}
\qquad\text{and}\qquad
\sigma_N(\mu_X) \ge \frac{\gamma(\tau)}{z-\sqrt{\beta(\tau)}},
\]
then
\[
\mathbb{P}\!\left( f(\mu_X) \in
\bigl[\nu_N(\mu_X)\pm z\,\sigma_N(\mu_X)\bigr]
\,\middle|\,\mathcal D_N
\right)\ge 1-\delta.
\]
\end{prop}

This leads to larger posterior variance where input uncertainty is high, increasing the width of credible intervals appropriately. The statement is Bayesian: under the well-specified GP model, the posterior credible interval is conservative whenever the variance lower bound above holds.

Sliced Wasserstein kernels \eqref{eq:slice_kernel} are positive definite, but in practice are implemented
using a finite number of random projections, which introduces Monte Carlo approximation error.
By contrast, PWA/\PCPWA are deterministic projections and avoid this additional randomness.

\begin{prop}\label{prop:swgp_bound}
Fix a direction $u\in\mathbb S^{d-1}$ and assume the analogue of Corollary~\ref{cor:unif_bound_gp} holds for the 1D projected model with parameters $(\tau,\delta)$: conditional on $\mathcal D_N$,
\[
|f(\mu)-\nu_N^u(\mu)| \le \sqrt{\beta_u(\tau)}\,\sigma_N^u(\mu) + \gamma_u(\tau)
\quad\text{for all }\mu\in\mathcal P
\]
with posterior probability at least $1-\delta$. Fix $\alpha\in(0,1)$ and $z:=z_{1-\alpha/2}$. If
\[
z > \sqrt{\beta_u(\tau)}
\qquad\text{and}\qquad
\sigma_N^u(\mu_X) \ge \frac{\gamma_u(\tau)}{z-\sqrt{\beta_u(\tau)}},
\]
then
\[
\mathbb{P}\!\left( f(\mu_X) \in
\bigl[\nu_N^u(\mu_X)\pm z\,\sigma_N^u(\mu_X)\bigr]
\,\middle|\,\mathcal D_N
\right)\ge 1-\delta.
\]
\end{prop}

\section{Experiments and Applications}

\label{sec:experiments}

Practitioners who model input uncertainty with Gaussian summaries often use a
covariance-only input map $x\mapsto \Sigma_x$, so that each input location
$x$ is represented as $X\sim\mathcal N(\mu_x,\Sigma_x)$ and only the covariance
$\Sigma_x$ enters the kernel \citep{morenomunoz2018heterogeneous}. 
This approach is efficient when noise is well‐approximated by local Gaussian perturbations but cannot represent skew, multimodality, or complex support geometry. Though we noticed the recent important work \citep{song2023frechet} tackles the input uncertainty by modeling in the space of covariates, it is not completely clear to us how their low-rank linear models can be compared to non-linear models and thus we did not include their method in the comparison. 
We cannot use histogram discretizations like the previous Wasserstein GP literature \citep{meunier2022slicedwasserstein,bachoc2017gaussian}, since the bin-width selection has major effect on the subsequent model. 
We evaluate our proposed \PWAGP (and \PCPWA where noted) on synthetic and real datasets using: (1) RMSE for posterior‐mean prediction accuracy; (2) CRPS \citep{arnold2024decompositions},
which jointly assesses predictive precision and sharpness, and (3) the probability coverage, which measures how well the posterior uncertainty covers the measured test observations.
\Cref{sec:eval-mtrics} contains the definitions of these metrics.
\subsection{Simulated distributional regression experiments}
\label{sec:simulated}

We study several synthetic distributional regression problems designed to
probe errors-in-variables (EIV) behavior and uncertainty quantification.
In all cases, each covariate is a finite cloud
$U_i=\{u_{ij}\}_{j=1}^{n_i}\subset\mathbb R^d$ with associated response
$Y_i$, and the task is to predict $Y$ for new clouds. We report test RMSE,
empirical coverage of nominal $90\%$ predictive intervals, and CRPS for
Regular GP on empirical means, Aggregated GP on individual samples, Wasserstein
GPs (\WGP, \SWGP, \PWAGP, \PCPWA), uncertain-input GP \UIGP \citet{girard2002gaussian}, and KME/MMD
distributional GPs \citep{meunier2022slicedwasserstein}.

\paragraph{1D scenarios.}
We consider three one-dimensional settings with Gaussian output noise
$\eta_i\sim\mathcal N(0,0.05^2)$.
(i) \textbf{1D-EIV}: latent covariates $x_i$ are evenly spaced on
$[0.05,0.95]$. We observe $U_i=\{x_i+\varepsilon_{ij}\}_{j=1}^{n_i}$ with
heteroscedastic Gaussian measurement noise $\varepsilon_{ij}\sim
\mathcal N(0,\sigma_i^2)$, where $\sigma_i^2$ increases with $x_i$, and set
$Y_i=f(x_i)+\eta_i$ for a nonlinear oscillatory $f$. This is a canonical
errors-in-variables problem. (ii) \textbf{1D-Var}:
$U_i$ consists of $u_{ij}\sim\mathcal N(\mu_i,\sigma_i^2)$ with varying
variances, and the response depends on both mean and spread,
$Y_i=\sin(2\pi\mu_i)+0.5\sigma_i^2+\eta_i$. (iii) \textbf{1D-Skew}:
inputs are log-normal, $u_{ij}\sim\exp(\mathcal N(m_i,s_i^2))$, and the
response depends on an inter-quantile range and location,
$Y_i=[Q_{0.8}(U_i)-Q_{0.2}(U_i)]+0.3\sin(m_i)+\eta_i$, where $Q_q(U_i)$ denotes the $q$-quantile of the input distribution (or,
equivalently in the simulations, of the empirical cloud $U_i$). 
1D-Var and 1D-Skew
thus require sensitivity to higher-order distributional features.

\paragraph{2D scenarios.}
We next consider two two-dimensional problems. (iv) \textbf{2D-mean}
(Gaussian location functional): for each $i$ we draw
$\mu_i\in[0.1,1]$ and sample $u_{ij}\sim\mathcal N(\mu_i\mathbf 1_2,
\mathrm{diag}(\sigma_{i1}^2,\sigma_{i2}^2))$ with moderate variances.
The response is a smooth function of the mean of $U_i$, averaged over
coordinates,
\[
Y_i = \frac{1}{n_i}\sum_{j=1}^{n_i}\bigl[\sin(u_{ij})+2e^{u_{ij}}\bigr]
+ \eta_i,
\]
so the problem is essentially a Gaussian location model.
(v) \textbf{2D-aniso-PC} (anisotropic rotated subspace): latent scalars
$z_i\in[0.05,0.95]$ generate samples in a rotated basis
$(z_{ij}^{\parallel},z_{ij}^{\perp})$ via
$z_{ij}^{\parallel}=z_i+\epsilon_{ij}^{\parallel}$,
$z_{ij}^{\perp}=\epsilon_{ij}^{\perp}$ with
$\epsilon_{ij}^{\parallel}\sim\mathcal N(0,\sigma_\parallel^2(z_i))$ and
$\epsilon_{ij}^{\perp}\sim\mathcal N(0,\sigma_\perp^2)$, where
$\sigma_\parallel^2(z_i)$ increases with $z_i$ and $\sigma_\perp^2$ is
small. A fixed $45^\circ$ rotation $R$ maps
$u_{ij}=R[z_{ij}^{\parallel},z_{ij}^{\perp}]^\top$, and the response
depends only on $z_i$,
$Y_i=\sin(4\pi z_i)+0.5 z_i+\eta_i$. Here the intrinsic variation lies
along a rotated one-dimensional subspace, so PCA-based kernels such as
\PCPWA are expected to help.

\paragraph{High-dimensional scenarios.}
Finally, we construct high-dimensional problems from the Ackley function.
(vi) \textbf{HD-Ackley-5D} and (vii) \textbf{HD-Ackley-10D}: latent
covariates $x_i\in[-2,2]^d$ with $d\in\{5,10\}$ are drawn uniformly,
we observe noisy clouds $u_{ij}=x_i+\varepsilon_{ij}$ with
$\varepsilon_{ij}\sim\mathcal N(0,0.1^2 I_d)$, and set
$Y_i=f_{\text{Ackley}}(x_i)+\eta_i$. These scenarios test how methods scale
with dimension in the presence of input noise and nontrivial curvature.

Table~\ref{tab:simulated} reports the resulting RMSE, coverage, and CRPS.
In 1D-EIV, Regular GP achieves RMSE $0.309$ but under-covers (65\%),
Aggregated GP is both inaccurate and grossly over-confident (5\% coverage),
while \WGP, \PWAGP, \PCPWA, and \SWGP all attain similar RMSE
($\approx0.30$) with coverage $\approx0.97$ and the lowest CRPS
($\approx0.17$), matching our EIV coverage analysis. In 1D-Var and
1D-Skew, mean- or moment-based baselines (Regular, UI-GP, KME/MMD) either
lose accuracy or under-cover, whereas Wasserstein kernels maintain
competitive RMSE and near-nominal coverage. In 1D-Var the Regular GP has
the best RMSE/CRPS, but the Wasserstein methods are within about 10-15\%
in RMSE; in 1D-Skew, MMD-GP attains the lowest RMSE and CRPS but with
coverage 0.58, whereas the Wasserstein models trade a slightly higher RMSE
for better-calibrated intervals.

In the benign 2D-mean scenario, uncertain-input and KME/MMD GPs-whose
assumptions match the Gaussian location structure-achieve the smallest RMSE
(0.0004-0.0026) with near-nominal coverage and CRPS, while Wasserstein
kernels remain reasonable but conservative (e.g. PWA-GP RMSE 0.0462,
coverage 1.0). This shows that our approach does not artificially dominate
when the problem is effectively Euclidean since WGP becomes regular GP when $\sigma_{i1}=\sigma_{1}$ and $\sigma_{i2}=\sigma_{2}$.
In contrast, in the anisotropic
2D-aniso-PC setting, methods based on Gaussian moments or mean embeddings
deteriorate markedly (UI-GP and KME-GP RMSE $\approx0.69$, coverage 0.15),
whereas \WGP, PWA-GP, \PCPWA, and MMD-GP all keep RMSE around 0.20 and
coverage between 0.65 and 0.98. \PCPWA in particular combines good accuracy
(RMSE 0.2036, very close to the best 0.1972) with high coverage (0.95) and
low CRPS, benefiting from aligning the kernel with principal directions of
the empirical clouds.

The high-dimensional Ackley experiments further probe robustness under
complex geometry. In 5D, Wasserstein kernels (especially PWA-GP and \WGP)
achieve lower RMSE and CRPS than Regular and UI-GP while maintaining higher
coverage, and Aggregated GP again fails badly. In 10D, all methods become
challenged, and MMD-GP attains the best RMSE/CRPS with reasonably high
coverage; nevertheless, Wasserstein kernels remain very close in RMSE (e.g. PWA-GP
and \WGP within roughly $10\%$ of the best) and retain the second best coverage.

Across all scenarios, sliced-\WGP (\SWGP) is competitive in 1D but its
optimization is numerically unstable in several higher-dimensional cases,
as indicated by the extremely large RMSE/CRPS values in Table~\ref{tab:simulated}.
Overall, PWA-GP and \PCPWA consistently deliver well-calibrated predictive
uncertainty in EIV settings while retaining competitive RMSE, and they
remain robust as the input distributions become heteroscedastic, skewed,
anisotropic, and high-dimensional, even when they are not the single
best-performing method in raw RMSE.

\begin{table}[t!]
\centering

\label{tab:simulated}
\resizebox{\textwidth}{!}{%
\begin{tabular}{llrrrrrrrrr}
\toprule
Scenario & Metric
& Reg & Agg & \WGP & \SWGP & PWA & \PCPWA & UI & KME & MMD \\
\midrule
2D-mean
  & RMSE & 8.500e-03 & 1.061e-01 & 5.280e-02 & 1.015e+07 & 4.620e-02 & 1.615e-01 & 4.000e-04 & 2.100e-03 & 2.600e-03 \\
  & Cov. & 4.250e-01 & 1.000e-01 & \textbf{1.000e+00} & 0.000e+00 & 6.250e-01 & \textbf{9.750e-01} & \textbf{1.000e+00} & \textbf{1.000e+00} & \textbf{9.750e-01} \\
  & CRPS & 5.100e-03 & 5.940e-02 & 4.540e-02 & 7.752e+06 & 3.970e-02 & 1.108e-01 & 2.000e-03 & 8.900e-03 & 1.010e-02 \\
\midrule
1D-EIV
  & RMSE & 3.090e-01 & 5.712e-01 & 3.014e-01 & 2.995e-01 & 3.014e-01 & 3.014e-01 & 3.462e-01 & 4.179e-01 & 3.029e-01 \\
  & Cov. & 6.500e-01 & 5.000e-02 & \textbf{9.670e-01} & \textbf{9.670e-01} & \textbf{9.670e-01} & \textbf{9.670e-01} & 5.670e-01 & 5.330e-01 & 6.670e-01 \\
  & CRPS & 1.848e-01 & 4.791e-01 & 1.716e-01 & 1.722e-01 & 1.716e-01 & 1.716e-01 & 2.179e-01 & 2.844e-01 & 1.775e-01 \\
\midrule
1D-Var
  & RMSE & 1.612e-01 & 4.810e-01 & 1.849e-01 & 1.819e-01 & 1.849e-01 & 1.849e-01 & 2.016e-01 & 2.697e-01 & 1.743e-01 \\
  & Cov. & 7.670e-01 & 0.000e+00 & \textbf{9.830e-01} & \textbf{9.830e-01} & \textbf{9.830e-01} & \textbf{9.830e-01} & 6.330e-01 & 6.170e-01 & 7.670e-01 \\
  & CRPS & 9.150e-02 & 3.648e-01 & 1.070e-01 & 1.058e-01 & 1.070e-01 & 1.070e-01 & 1.135e-01 & 1.552e-01 & 9.900e-02 \\
\midrule
1D-Skew
  & RMSE & 1.768e-01 & 4.779e-01 & 1.629e-01 & 1.633e-01 & 1.629e-01 & 1.629e-01 & 1.762e-01 & 1.892e-01 & 1.406e-01 \\
  & Cov. & 4.330e-01 & 0.000e+00 & \textbf{9.500e-01} & \textbf{9.670e-01} & \textbf{9.500e-01} & \textbf{9.500e-01} & 4.170e-01 & 4.500e-01 & 5.830e-01 \\
  & CRPS & 1.075e-01 & 3.618e-01 & 9.730e-02 & 9.750e-02 & 9.730e-02 & 9.730e-02 & 1.057e-01 & 1.152e-01 & 8.000e-02 \\
\midrule
2D-aniso-PC
  & RMSE & 1.972e-01 & 5.435e-01 & 2.086e-01 & 2.541e+06 & 2.130e-01 & 2.036e-01 & 6.893e-01 & 6.893e-01 & 2.009e-01 \\
  & Cov. & 6.750e-01 & 2.500e-02 & \textbf{9.750e-01} & 0.000e+00 & \textbf{9.500e-01} & \textbf{9.500e-01} & 1.500e-01 & 1.500e-01 & 6.500e-01 \\
  & CRPS & 1.086e-01 & 3.846e-01 & 1.252e-01 & 2.581e+06 & 1.211e-01 & 1.124e-01 & 5.528e-01 & 5.528e-01 & 1.098e-01 \\
\midrule
HD-Ackley-5D
  & RMSE & 6.071e-01 & 2.483e+00 & 5.413e-01 & 4.080e+08 & 5.190e-01 & 5.562e-01 & 6.062e-01 & 6.574e-01 & 5.080e-01 \\
  & Cov. & 7.750e-01 & 2.500e-02 & 9.000e-01 & 0.000e+00 & 8.870e-01 & \textbf{9.630e-01} & 7.750e-01 & 7.750e-01 & 7.870e-01 \\
  & CRPS & 3.501e-01 & 1.879e+00 & 3.302e-01 & 3.081e+07 & 2.715e-01 & 3.147e-01 & 3.495e-01 & 3.607e-01 & 2.969e-01 \\
\midrule
HD-Ackley-10D
  & RMSE & 5.925e-01 & 2.487e+00 & 5.873e-01 & 5.888e+00 & 6.003e-01 & 6.110e-01 & 5.962e-01 & 6.112e-01 & 5.450e-01 \\
  & Cov. & 2.630e-01 & 7.500e-02 & 6.880e-01 & 0.000e+00 & 7.000e-01 & 7.620e-01 & 6.750e-01 & 6.250e-01 & 7.630e-01 \\
  & CRPS & 4.138e-01 & 1.887e+00 & 3.492e-01 & 2.431e+00 & 3.567e-01 & 3.718e-01 & 3.470e-01 & 3.552e-01 & 3.309e-01 \\
\bottomrule
\end{tabular}}
\caption{Simulated distributional regression experiments. For each scenario
we report test RMSE, empirical coverage of nominal $90\%$ predictive
intervals (Cov.), and CRPS. ``Reg'' = Regular GP on empirical means,
``Agg'' = Aggregated GP, ``UI'' = uncertain-input GP.
Extremely large \SWGP values indicate numerical divergence. Values of Cov.\ above $0.9$ are shown in \textbf{bold}.}
\end{table}

\subsection{Calibration of particle accelerators}
\label{sec:online_examples}

Particle accelerators play a central role in scientific discovery and
industrial applications. In a typical machine, focusing elements such as
quadrupole magnets guide the charged particle beam and preserve beam
quality, measured by the transverse emittance at the end of the accelerator.
Here we study the relationship between the strengths of five quadrupole
magnets in a linac section and the resulting beam emittance.

For each of the $N=500$ machine settings we observe a five-dimensional vector
of magnet strengths and two emittances (horizontal and vertical). Each magnet
is specified with a relative tolerance ($5\%$ or $0.5\%$) around its nominal
strength \cite{tao2018}, so the “input’’ available to a regression model is a noisy proxy
$U_i$ of an unknown true setting $X_i$-a classical errors-in-variables
situation. We therefore model each setting as
\[
U_i \sim \mathcal N(\mu_i,\Sigma_i),
\qquad
\Sigma_i = \mathrm{diag}(\mu_i \odot \varepsilon),
\]
where $\mu_i\in\mathbb R^5$ is the magnet vector for shot $i$ and
$\varepsilon\in\{0.05,0.005\}$ is the relative error. From each
$\mathcal N(\mu_i,\Sigma_i)$ we draw a small cloud of samples to represent
$U_i$, and define the scalar response
$Y_i = (\text{horizontal emittance}) \times (\text{vertical emittance}) /
10^{-12}$. We use an 80/20 train-test split.

We fit both Euclidean and
distributional GP models. Euclidean baselines are standard GPs on empirical
means $m_i=\mathbb E[U_i]$ with RBF, Matérn $3/2$, Matérn $5/2$ and
exponential kernels. Distributional competitors are our Wasserstein GPs
(PWA-GP, \WGP, \PCPWA) that act on the sample clouds using separable, full and
PCA-based Wasserstein kernels, together with the uncertain-input GP
\UIGP of \citet{girard2002gaussian}.

Table~\ref{tab:accel} reports the results for both $5\%$ and $0.5\%$
relative errors. All methods (except \SWGP) nearly interpolate the training
data (train RMSE $<2\times10^{-2}$, not shown). On the test set, UI-GP
consistently achieves the best accuracy and calibration: RMSE $0.1340$ with
coverage $0.86$ and CRPS $0.0557$ at $5\%$ error, improving to RMSE $0.1294$,
coverage $0.92$ and CRPS $0.0441$ at $0.5\%$ error-very close to the nominal
$90\%$ target. Euclidean GPs on means have larger RMSE ($0.22$-$0.24$ at
$5\%$, $0.22$-$0.28$ at $0.5\%$) and higher CRPS, but still decent coverage
($0.91$-$0.98$), reflecting that in this nearly Gaussian, low-dimensional
setting a moment-based EIV correction already works well.

Our Wasserstein GPs (PWA-GP, \WGP, \PCPWA) are broadly comparable to the
mean-based Euclidean GPs on this dataset: they achieve similar qualitative
behavior and conservative coverage ($0.96$-$0.98$), but do not provide a
systematic improvement in point accuracy or CRPS. This is consistent with
the structure of the input uncertainty here, each 
$U_i$ is close to Gaussian
and its dispersion is largely determined by the nominal setting $\mu_i$ via a
fixed relative tolerance. Therefore, much of the useful information is already
captured by the empirical mean representation. The \SWGP\ is numerically
unstable for both error levels, with diverging RMSE and CRPS and zero
coverage, and is therefore not a viable competitor on this dataset.

Overall, the accelerator results align with our simulations results in Table \ref{tab:simulated} (1D-EIV):
when the input noise is close to Gaussian and the response depends mainly on
low-order moments, mean-based and uncertain-input GPs can match (and in this
case outperform) Wasserstein GPs, whereas the latter show clear benefits in
the more non-Gaussian and anisotropic EIV scenarios of
Section~\ref{sec:simulated}. 

\begin{table}[t!]
\centering

\label{tab:accel}
\resizebox{.85\textwidth}{!}{%
\begin{tabular}{lcccccc}
\toprule
& \multicolumn{3}{c}{5\% error} & \multicolumn{3}{c}{0.5\% error} \\
Method
& RMSE & Cov. & CRPS
& RMSE & Cov. & CRPS \\
\midrule
RBF GP            & 0.2186 & 0.910 & 0.0883 & 0.2177 & 0.920 & 0.0806 \\
Matérn $3/2$ GP   & 0.2351 & 0.930 & 0.1010 & 0.2378 & 0.930 & 0.0997 \\
Matérn $5/2$ GP   & 0.2228 & 0.940 & 0.0912 & 0.2242 & 0.940 & 0.0892 \\
Exp.\ GP          & 0.2432 & 0.980 & 0.1380 & 0.2847 & 0.980 & 0.1392 \\
PWA-GP            & 0.2798 & 0.980 & 0.1690 & 0.3373 & 0.980 & 0.1681 \\
\SWGP              & 2.7234$\times 10^{6}$ & 0.000 & 2.1404$\times 10^{6}$
                  & 3.1142$\times 10^{6}$ & 0.000 & 3.9275$\times 10^{6}$ \\
\WGP               & 0.3169 & 0.970 & 0.1949 & 0.2806 & 0.970 & 0.1401 \\
\PCPWA             & 0.3072 & 0.960 & 0.1682 & 0.2828 & 0.970 & 0.1430 \\
Uncertain-input GP & 0.1340 & 0.860 & 0.0557 & 0.1294 & 0.920 & 0.0441 \\
\bottomrule
\end{tabular}}
\caption{Accelerator dataset: test RMSE, empirical coverage of nominal
$90\%$ predictive intervals (Cov.), and CRPS for two relative error levels.
Euclidean GPs use empirical means as inputs; distributional GPs act on full
input distributions.}
\end{table}

\subsection{Noisy Trajectories from NOAA Drifters}

This experiment compares GP-RBF and \WGP-RBF in predicting mean temperatures, emphasizing input uncertainty. We evaluate models using RMSE and CRPS metrics on NOAA's Global Drifter Program dataset, containing hourly drifter trajectories and sea surface temperatures (SST) observations  \citep{elipot2016global,elipot2022dataset}. 
While the SST measurements are collected over space at a high temporal frequency, scientific interest is often focused on the time-aggregated behavior of SSTs to understand large-scale modes of ocean variability \citep{Philander1983, Mantua2002, Schlesinger1994, Thompson2000}. Therefore, we consider these trajectories as distributional input models, where the distribution of inputs represents the collection of geospatial positions corresponding to the hourly measurements. 
The data is preprocessed by selecting trajectories with low temperature variance ($Y_{\text{var\_temp}}$). Each collection of scattered points on a trajectory is then represented by a 2D normal distribution with empirical mean position 
and empirical covariance matrix. We have to emphasize that the trajectories of drifters (i.e., floating buoys on the oceanic surface) are highly noisy, self-intersecting and therefore usual functional data representation is not appropriate here. 
Therefore, we consider these trajectories as distributional input models, where both measurement noise and ocean variability-information GP-RBF ignores by using only point estimates. The response ($Y_{\text{mean\_temp}}$) is the mean SST. We train on one year's data (e.g., 1999) and test on another year (e.g., 1998), initially focusing on 100 trajectories with lowest variance to ensure valid scalar responses.


\begin{figure}
\centering 
\includegraphics[width=0.49\textwidth,height=3.75cm]{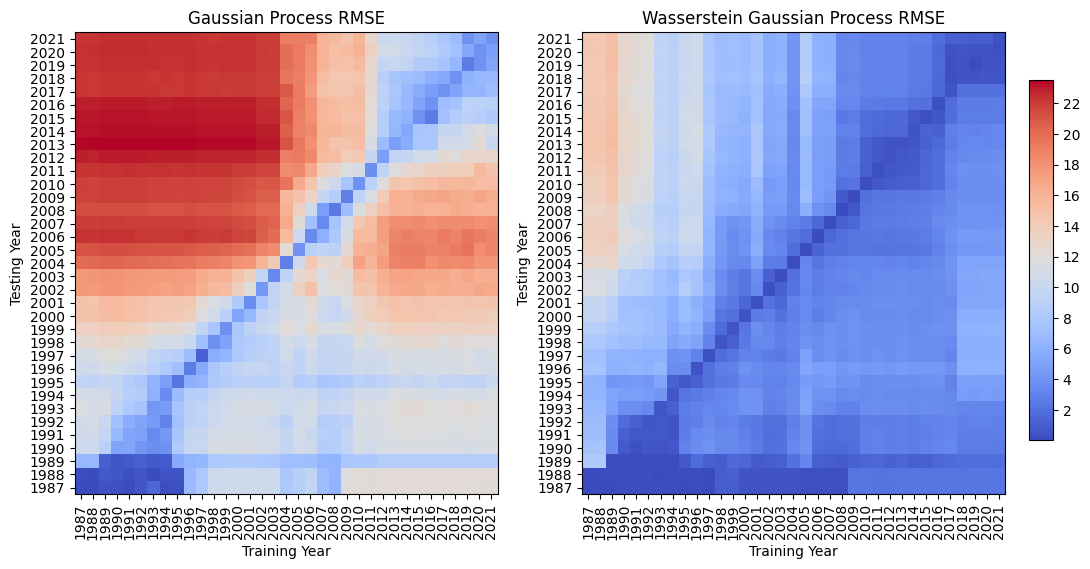}
\includegraphics[width=0.49\textwidth,height=3.75cm]{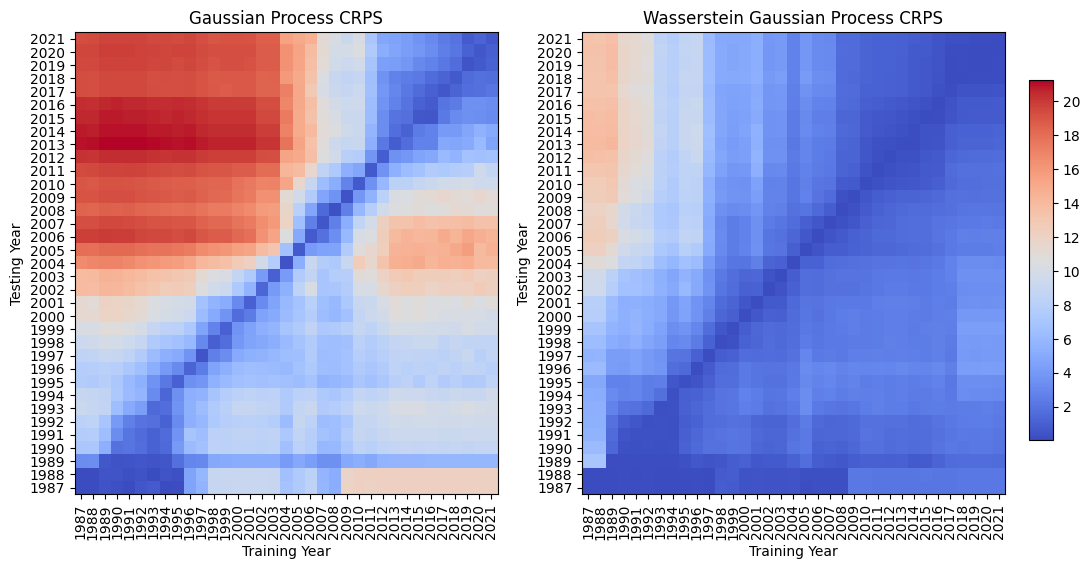}
\caption{RMSE (left) and CRPS (right) for GP-RBF vs.\ \WGP-RBF across training/testing years 1987-2022. GP-RBF degrades when trained on recent and tested on older years, indicating poor generalization under temporal shifts. \WGP-RBF remains robust, consistently achieving lower errors and better probabilistic calibration.}
\label{fig:noaagdp-2} 
\end{figure}

We then evaluated both models on every pair of training/testing years (1987-2022). The sliced \WGP takes much higher memory and is not considered here. 
Figure~\ref{fig:noaagdp-2} shows RMSE and CRPS heatmaps. GP-RBF suffers high errors off the diagonal-overfitting recent data and failing to extrapolate to earlier distributions. In contrast, \WGP-RBF maintains low RMSE and CRPS across all year combinations, thanks to its explicit treatment of input uncertainty. The diagonal structure reflects optimal performance when training and testing distributions match, but \WGP-RBF’s stronger off-diagonal performance underscores its superior transferability for applications with nonstationary input variability, such as climate monitoring.

\section{Discussion}
\label{sec:discussion}
We have argued that uncertainty quantification, rather than purely pointwise accuracy, should guide the design of GP models when inputs are measured with error. By lifting inputs to probability measures and endowing this space with Wasserstein geometry, PWA-based GPs provide a simple, deterministic alternative to latent-input and sampling-based EIV approaches. The theory now separates a deterministic net-extension bound from a model-based GP posterior statement: under Lipschitz and covering-number conditions, one obtains a uniform high-probability posterior band and a simple sufficient condition for conservative credible intervals. Our experiments confirm that these ideas are consistent with substantially better-calibrated predictive intervals on real data.
Our empirical findings suggest that classical Euclidean kernels are often overly confident when covariates are noisy, whereas Wasserstein-type kernels and PWA in particular produce credible intervals that track the true variability across different noise regimes and non-Gaussian input distributions. This makes \PWAGP especially appealing in scientific and engineering applications, such as accelerator calibration and oceanographic monitoring, where decisions are driven by the reliability of uncertainty quantification.

Future work includes extending the analysis to high-dimensional, non-Gaussian noise \citep{tekriwal2026ggmpsgeneralizedgaussianmixture} and non-Euclidean input spaces \citep{luo2026wavelet}, incorporating learned or problem-specific transport costs.

\section*{Acknowledgments}
This research was supported in part by the U.S. Department of Energy, Office of Science, Office of Advanced Scientific Computing Research's Applied Mathematics program under Contract No. DE-AC02-05CH11231 at Lawrence Berkeley National Laboratory.

\clearpage{}

\bibliographystyle{apalike}
\bibliography{TMLR}

\newpage
\appendix

\numberwithin{figure}{section} \numberwithin{table}{section}
\section{Evaluation metrics}
\label{sec:eval-mtrics}
We report three complementary metrics for predictive performance. Let
$\widehat m_i$ and $\widehat s_i$ denote the posterior mean and posterior
standard deviation for the $i$th test input, and let $y_i$ be the
corresponding test response.

First, the root mean squared error (RMSE) measures point-prediction accuracy:
\[
\operatorname{RMSE}
=
\left(
\frac{1}{n_{\mathrm{test}}}
\sum_{i=1}^{n_{\mathrm{test}}}
(\widehat m_i-y_i)^2
\right)^{1/2}.
\]

Second, for a nominal $100(1-\alpha)\%$ predictive interval
\[
C_i
=
\bigl[
\widehat m_i - z_{1-\alpha/2}\widehat s_i,\,
\widehat m_i + z_{1-\alpha/2}\widehat s_i
\bigr],
\]
the empirical coverage is
\[
\operatorname{Cov}_{1-\alpha}
=
\frac{1}{n_{\mathrm{test}}}
\sum_{i=1}^{n_{\mathrm{test}}}
\mathbf{1}\{y_i \in C_i\}.
\]
Coverage below the nominal level is called \emph{undercoverage} and indicates
that the predictive intervals are too narrow, i.e., the model is
overconfident.

Third, we report the continuous ranked probability score (CRPS), a proper
scoring rule for the full predictive distribution:
\[
\operatorname{CRPS}(F_i,y_i)
=
\int_{-\infty}^{\infty}
\bigl(F_i(t)-\mathbf{1}\{y_i \le t\}\bigr)^2\,dt,
\]
where $F_i$ is the predictive cumulative distribution function for the $i$th
test point. Lower CRPS is better; it rewards both calibration and sharpness.
\section{Proof for Proposition \ref{ref:descriptive_prop}}
\begin{proof}
With $f(x)=c+w^\top x$, we have
\[
Y = f(X+\varepsilon_X)+\varepsilon = f(X) + w^\top \varepsilon_X + \varepsilon.
\]
Thus $f(X)-Y = -(w^\top\varepsilon_X+\varepsilon)$ and conditional on $X$,
\[
w^\top\varepsilon_X+\varepsilon \sim \mathcal{N}(0,\;w^\top\Sigma_X w+\sigma^2).
\]
Therefore,
\[
\mathbb{P}(f(X)\in I_\alpha\mid X)
=
\mathbb{P}\left(|w^\top\varepsilon_X+\varepsilon|\le z_{1-\alpha/2}\sigma\right)
=
2\Phi\!\left(\frac{z_{1-\alpha/2}\sigma}{\sqrt{\sigma^2+w^\top\Sigma_X w}}\right)-1,
\]
which is strictly less than $1-\alpha$ whenever $w^\top\Sigma_X w>0$.
\end{proof}

\section{Proof for Theorem \ref{thm:unif_bound}}
\label{sec:proof_main_thm}

\begin{proof}
We first verify the Lipschitz continuity of $\nu_N$. Since
\[
\nu_N(\mu)=\sum_{i=1}^N \alpha_i k(\mu,\mu_i),
\]
and $k$ is $L_k$-Lipschitz in its first argument with respect to $W_1$,
\[
|\nu_N(\mu)-\nu_N(\mu')|
\le \sum_{i=1}^N |\alpha_i|\,|k(\mu,\mu_i)-k(\mu',\mu_i)|
\le L_k \Bigl(\sum_{i=1}^N |\alpha_i|\Bigr) W_1(\mu,\mu').
\]
Because $\alpha=(K+\sigma_*^2 I_N)^{-1}y$ and $K+\sigma_*^2 I_N$ is invertible, $\|\alpha\|_\infty<\infty$, hence
\[
\sum_{i=1}^N |\alpha_i|\le N\|\alpha\|_\infty
\quad\text{and therefore}\quad
|\nu_N(\mu)-\nu_N(\mu')|\le L_{\nu_N} W_1(\mu,\mu').
\]

Now let $k_\mu=(k(\mu,\mu_1),\ldots,k(\mu,\mu_N))^\top$ and set $A=(K+\sigma_*^2 I_N)^{-1}$. Then
\[
\sigma_N^2(\mu)=k(\mu,\mu)-k_\mu^\top A k_\mu.
\]
By assumption, $\sigma_N$ has modulus of continuity $\omega_{\sigma_N}$, i.e.
\[
|\sigma_N(\mu)-\sigma_N(\mu')|\le \omega_{\sigma_N}(\tau)
\qquad\text{whenever }W_1(\mu,\mu')\le \tau.
\]

Let $\{\bar\mu_1,\ldots,\bar\mu_M\}$ be a $\tau$-net of $\mathcal P$. Fix any $\mu\in\mathcal P$ and choose $j$ such that $W_1(\mu,\bar\mu_j)\le \tau$. Then
\[
|f(\mu)-\nu_N(\mu)|
\le |f(\mu)-f(\bar\mu_j)|
+|f(\bar\mu_j)-\nu_N(\bar\mu_j)|
+|\nu_N(\bar\mu_j)-\nu_N(\mu)|.
\]
The first and third terms are bounded by Lipschitz continuity:
\[
|f(\mu)-f(\bar\mu_j)|\le L_f\tau,
\qquad
|\nu_N(\bar\mu_j)-\nu_N(\mu)|\le L_{\nu_N}\tau.
\]
By hypothesis,
\[
|f(\bar\mu_j)-\nu_N(\bar\mu_j)|\le B\,\sigma_N(\bar\mu_j)
\le B\,[\sigma_N(\mu)+\omega_{\sigma_N}(\tau)].
\]
Combining the three bounds yields
\[
|f(\mu)-\nu_N(\mu)|
\le B\,\sigma_N(\mu)+(L_f+L_{\nu_N})\tau+B\,\omega_{\sigma_N}(\tau),
\]
as claimed.
\end{proof}

\begin{proof}[Proof of Corollary \ref{cor:unif_bound_gp}]
Let $\{\bar\mu_1,\ldots,\bar\mu_M\}$ be a $\tau$-net of $\mathcal P$ in $W_1$, where
$M = M(\tau,\mathcal P)$. Under the well-specified GP model, conditional on $\mathcal D_N$,
for each $j=1,\dots,M$,
\[
f(\bar\mu_j)\mid \mathcal D_N \sim
\mathcal N\!\bigl(\nu_N(\bar\mu_j),\,\sigma_N^2(\bar\mu_j)\bigr).
\]

Fix $j$. If $\sigma_N(\bar\mu_j)>0$, then
\[
Z_j
:=
\frac{f(\bar\mu_j)-\nu_N(\bar\mu_j)}{\sigma_N(\bar\mu_j)}
\;|\; \mathcal D_N
\sim \mathcal N(0,1),
\]
and therefore, by the definition
\[
\beta(\tau)
=
\left[
\Phi^{-1}\!\left(1-\frac{\delta}{2M(\tau,\mathcal P)}\right)
\right]^2,
\]
we have
\begin{align*}
\mathbb P\!\left(
|f(\bar\mu_j)-\nu_N(\bar\mu_j)|
> \sqrt{\beta(\tau)}\,\sigma_N(\bar\mu_j)
\,\middle|\, \mathcal D_N
\right)
&=
\mathbb P\!\left(
|Z_j|>\Phi^{-1}\!\left(1-\frac{\delta}{2M(\tau,\mathcal P)}\right)
\right) \\
&=
\frac{\delta}{M(\tau,\mathcal P)}.
\end{align*}
If instead $\sigma_N(\bar\mu_j)=0$, then the posterior law is degenerate and
$f(\bar\mu_j)=\nu_N(\bar\mu_j)$ almost surely given $\mathcal D_N$, so the same inequality
holds with probability $0\le \delta/M(\tau,\mathcal P)$.

Hence, for every $j=1,\dots,M$,
\[
\mathbb P\!\left(
|f(\bar\mu_j)-\nu_N(\bar\mu_j)|
> \sqrt{\beta(\tau)}\,\sigma_N(\bar\mu_j)
\,\middle|\, \mathcal D_N
\right)
\le
\frac{\delta}{M(\tau,\mathcal P)}.
\]
Applying the union bound over the $M(\tau,\mathcal P)$ net points yields
\[
\mathbb P\!\left(
|f(\bar\mu_j)-\nu_N(\bar\mu_j)|
\le
\sqrt{\beta(\tau)}\,\sigma_N(\bar\mu_j)
\text{ for all } j=1,\dots,M
\,\middle|\, \mathcal D_N
\right)
\ge 1-\delta.
\]

On this event, Theorem~\ref{thm:unif_bound} applies with $B=\sqrt{\beta(\tau)}$, giving
\[
|f(\mu)-\nu_N(\mu)|
\le
\sqrt{\beta(\tau)}\,\sigma_N(\mu)
+
(L_f+L_{\nu_N})\tau
+
\sqrt{\beta(\tau)}\,\omega_{\sigma_N}(\tau)
\qquad
\text{for all }\mu\in\mathcal P.
\]
Since
\[
\gamma(\tau)=(L_f+L_{\nu_N})\tau+\sqrt{\beta(\tau)}\,\omega_{\sigma_N}(\tau),
\]
the claimed bound follows.
\end{proof}

\begin{proof}[Proof of Corollary \ref{cor:unif_bound_pcpwa}]
Let $\{\bar\mu_1,\ldots,\bar\mu_M\}$ be a $\tau$-net of $(\mathcal P,D)$, where
$M=M_{\mathrm{PCPWA}}(\tau,\mathcal P)$. Under the well-specified GP model with prior
$f\sim GP(0,k_{\mathrm{PCPWA}})$, conditional on $\mathcal D_N$, for each $j=1,\dots,M$,
\[
f(\bar\mu_j)\mid \mathcal D_N
\sim
\mathcal N\!\bigl(\nu_N(\bar\mu_j),\,\sigma_N^2(\bar\mu_j)\bigr).
\]

Define
\[
\beta_{\mathrm{PCPWA}}(\tau)
=
\left[
\Phi^{-1}\!\left(
1-\frac{\delta}{2M_{\mathrm{PCPWA}}(\tau,\mathcal P)}
\right)
\right]^2.
\]
Arguing exactly as in the proof of Corollary~\ref{cor:unif_bound_gp}, for each net point
$\bar\mu_j$,
\[
\mathbb P\!\left(
|f(\bar\mu_j)-\nu_N(\bar\mu_j)|
>
\sqrt{\beta_{\mathrm{PCPWA}}(\tau)}\,\sigma_N(\bar\mu_j)
\,\middle|\,\mathcal D_N
\right)
\le
\frac{\delta}{M_{\mathrm{PCPWA}}(\tau,\mathcal P)}.
\]
A union bound therefore gives, conditional on $\mathcal D_N$,
\[
|f(\bar\mu_j)-\nu_N(\bar\mu_j)|
\le
\sqrt{\beta_{\mathrm{PCPWA}}(\tau)}\,\sigma_N(\bar\mu_j),
\qquad j=1,\dots,M,
\]
with posterior probability at least $1-\delta$.

Now apply Theorem~\ref{thm:unif_bound} with the pseudo-metric $D$ in place of $W_1$,
using the assumed $D$-Lipschitz bounds for $f$ and $k_{\mathrm{PCPWA}}$ and the modulus
of continuity $\omega_{\sigma_N}$ with respect to $D$, with
\[
B=\sqrt{\beta_{\mathrm{PCPWA}}(\tau)}.
\]
This yields
\[
|f(\mu)-\nu_N(\mu)|
\le
\sqrt{\beta_{\mathrm{PCPWA}}(\tau)}\,\sigma_N(\mu)
+
(L_f+L_{\nu_N})\tau
+
\sqrt{\beta_{\mathrm{PCPWA}}(\tau)}\,\omega_{\sigma_N}(\tau)
\qquad
\text{for all }\mu\in\mathcal P,
\]
with posterior probability at least $1-\delta$. Writing
\[
\gamma_{\mathrm{PCPWA}}(\tau)
=
(L_f+L_{\nu_N})\tau
+
\sqrt{\beta_{\mathrm{PCPWA}}(\tau)}\,\omega_{\sigma_N}(\tau)
\]
completes the proof.
\end{proof}

\section{Proof for Proposition \ref{prop:swgp_bound}}
\begin{proof}
Fix $u\in\mathbb S^{d-1}$. Let $\mathcal E_u$ denote the posterior event that the projected uniform bound holds:
\[
|f(\mu)-\nu_N^u(\mu)| \le \sqrt{\beta_u(\tau)}\,\sigma_N^u(\mu) + \gamma_u(\tau)
\qquad\text{for all }\mu\in\mathcal P.
\]
By assumption, $\mathbb P(\mathcal E_u\mid\mathcal D_N)\ge 1-\delta$.

On $\mathcal E_u$, evaluating at $\mu=\mu_X$ gives
\[
|f(\mu_X)-\nu_N^u(\mu_X)| \le \sqrt{\beta_u(\tau)}\,\sigma_N^u(\mu_X) + \gamma_u(\tau).
\]
Assume $z>\sqrt{\beta_u(\tau)}$ and $\sigma_N^u(\mu_X)\ge \gamma_u(\tau)/(z-\sqrt{\beta_u(\tau)})$. Then
\[
\sqrt{\beta_u(\tau)}\,\sigma_N^u(\mu_X)+\gamma_u(\tau)\le z\,\sigma_N^u(\mu_X),
\]
so $|f(\mu_X)-\nu_N^u(\mu_X)| \le z\,\sigma_N^u(\mu_X)$, which is equivalent to $f(\mu_X)\in I^u_\alpha(\mu_X)$. Therefore,
\[
\mathbb P\!\left(f(\mu_X)\in I^u_\alpha(\mu_X)\,\middle|\,\mathcal D_N\right)
\ge \mathbb P(\mathcal E_u\mid\mathcal D_N) \ge 1-\delta,
\]
which proves the claim.
\end{proof}

\section{Bounds for derivatives of log likelihoods}

\label{appdx:derivatives-bounds}

Here we derive bounds of derivatives for $k_{GW_{p}}$ as below:
\begin{equation}
k_{GW_{p}}(\mu,\nu)=\lambda\exp(-\sigma GW_{p}(\mu,\nu)^{2}),\label{eq:kernel_GWp}
\end{equation}

\begin{lem}
\label{lem:(Bounds-for-derivatives}(Bounds for derivatives of log
likelihoods with GW) For $p=2$ and $\mu,\nu$ are both multivariate
Gaussian distributions defined on $\mathbb{R}^{d},\mathbb{R}^{m}$,
the bounds for derivatives with respect to the kernel hyper-parameters
$\lambda,\sigma$ are in closed form as functions of $GW_{p}(\mu,\nu)^{2}$. 
\end{lem}

We want to derive the $\frac{\partial\bm{K}}{\partial\lambda}$
and $\frac{\partial\bm{K}}{\partial\sigma}$ and plug it back respectively.
We want to use Theorem 3.1, 4.1 and Proposition 4.2 in \citet{delon2022gromov}
to obtain a lower bound for these derivatives as a closed form surrogate
to the exact. Suppose without any loss of generality that $d\leq m$.
Let $\mu=N(\mu_{0},\Sigma_{0})$ on $\mathbb{R}^{m}$ and $\nu=N(\mu_{1},\Sigma_{1})$
on $\mathbb{R}^{d}$ be two Gaussian measures. Let $P_{0},D_{0}$
and $P_{1},D_{1}$ be the respective diagonalizations of $\Sigma_{0}=P_{0}D_{0}P_{0}^{T}\in\mathbb{R}^{m\times m}$
and $\Sigma_{1}=P_{1}D_{1}P_{1}^{T}\in\mathbb{R}^{m\times m}$ which
sort eigenvalues in decreasing order. We suppose that $\Sigma_{0}$
is non-singular. Then the $GW_{2}(\mu,\nu)$ is bounded by 
\begin{equation}
LGW_{2}(\mu,\nu)^{2}\leq GW_{2}(\mu,\nu)^{2}\leq GGW_{2}(\mu,\nu)^{2},\label{eq:LGW_GGW}
\end{equation}
where the lower and upper bounds have the following closed form expressions:
\begin{align}
LGW_{2}(\mu,\nu)^{2} & =4\left(\text{tr}(D_{0})-\text{tr}(D_{1})\right)^{2}+4\left(\|D_{0}\|_{F}-\|D_{1}\|_{F}\right)^{2}+4\|D_{0}^{(d)}-D_{1}\|_{F}^{2}\\
 & +4\left(\|D_{0}\|_{F}-\|D_{0}^{(d)}\|_{F}\right)^{2}.\nonumber \\
GGW_{2}(\mu,\nu)^{2} & =4\left(\text{tr}(D_{0})-\text{tr}(D_{1})\right)^{2}+8\left(\|D_{0}^{(d)}-D_{1}\|_{F}^{2}\right)+8\left(\|D_{0}\|_{F}^{2}-\|D_{0}^{(d)}\|_{F}^{2}\right).
\end{align}
and the we use the notation $D^{(d)}$ to denote the $d\times d$
submatrix containing the first $d$ rows and the first $d$ columns
of any matrix $D$; $\|\cdot\|_{F}$ is the matrix Frobenius norm.
Recall that we can set $x\sim\mu=N(\mu_{0},\Sigma_{0})$ and $x'\sim\nu=N(\mu_{1},\Sigma_{1})$,
then we have: 
\begin{align*}
\frac{\partial\bm{K}}{\partial\lambda} & =\left\{ \exp(-\sigma GW_{2}(x,x')^{2})\right\} _{x,x'},\\
\frac{\partial\bm{K}}{\partial\sigma} & =\left\{ -\lambda GW_{2}(x,x')^{2}\exp(-\sigma GW_{2}(x,x')^{2})\right\} _{x,x'}.
\end{align*}
These expressions admit the bound \eqref{eq:LGW_GGW} due to the fact
that exponential function is strictly monotonic over the real line.
This is again a computational tool which can be used in optimization
algorithms to find the MLEs of $\lambda$ and $\sigma$. Instead of
putting nested loops to solve the optimal transportation problem in
the definition of $GW_{2}$, we only optimize the $LGW_{2}(\mu,\nu)^{2}$
and $GGW_{2}(\mu,\nu)^{2}$ (in closed form) to simplify the fitting
procedure. However, in practice it usually suffices to maximize $LGW_{2}(\mu,\nu)^{2}$.
At the same time, the \eqref{eq:LGW_GGW} can be used to verify the
Assumption 3 in \citet{zhou2019gaussian} when the $P_{0},D_{0}$
and $P_{1},D_{1}$ are parameterized by a smooth parameter (i.e.,
as a smooth matrix-valued function of a kernel parameter).
\section{Computational complexity of competing methods}
\label{sec:complexity}
We briefly compare the leading-order computational complexity of the GP
methods considered in this work.  For all models, once a kernel matrix
$K \in \mathbb R^{N\times N}$ has been assembled, exact GP training via
Cholesky decomposition costs $\mathcal O(N^3)$, and prediction for
$N_{\star}$ test inputs costs $\mathcal O(N^2 + N N_{\star})$ (or
$\mathcal O(N N_{\star})$ once the Cholesky factor is cached).  The main
differences between methods therefore arise from (i) the number of
\emph{effective} training points they use and (ii) the cost of computing
each kernel entry from the underlying clouds $U_i$.

Throughout, $U_i = \{u_{ij}\}_{j=1}^{m_i}\subset\mathbb R^d$ denotes the
empirical input distribution for training example $i$, and
$M = \max_i m_i$.  We assume $m_i \asymp M$ for simplicity.

\paragraph{Regular GP on means.}
The regular GP baseline \citep{williams2006gpml} treats each
distribution as a single point in $\mathbb R^d$ via its empirical mean
$\bar u_i = m_i^{-1}\sum_{j=1}^{m_i} u_{ij}$.  Computing all means costs
$\mathcal O(N M d)$ once, and a kernel evaluation between two means is
$\mathcal O(d)$.  Building $K$ therefore costs $\mathcal O(N^2 d)$, after
which GP training is $\mathcal O(N^3)$.  This is the reference complexity
for a standard Euclidean GP with $N$ inputs of dimension $d$.

\paragraph{Aggregated GP on raw samples.}
The aggregated GP baseline (an ensemble of standard GPs
\citep{williams2006gpml}) treats each sample $u_{ij}$ as an
independent input.  The effective number of training points becomes
$\tilde N = \sum_i m_i \asymp N M$.  Kernel evaluations remain
$\mathcal O(d)$, but the kernel matrix is now of size
$\tilde N \times \tilde N$, so training scales as
\[
\mathcal O(\tilde N^3) \;=\; \mathcal O(N^3 M^3),
\]
with kernel assembly cost $\mathcal O(N^2 M^2 d)$.  This cubic dependence
on $M$ makes the aggregated GP rapidly infeasible as the number of
samples per distribution grows.

\paragraph{Uncertain-input GP.}
The uncertain-input GP of \citet{girard2002gaussian} (see also
\citet{mchutchon2011gaussian}) replaces each $U_i$ by a Gaussian
$\mathcal N(\mu_i,\Sigma_i)$, where $\mu_i$ and $\Sigma_i$ are the
empirical mean and covariance.  Estimating these parameters costs
$\mathcal O(N M d^2)$ in general (or $\mathcal O(N M d)$ for diagonal
covariances).  The kernel between two Gaussians has a closed form
involving $(\Sigma_i + \Sigma_j)$ and $(\mu_i - \mu_j)$; with full
covariances this requires an inversion and determinant per pair, for a
cost of $\mathcal O(d^3)$ per kernel entry and $\mathcal O(N^2 d^3)$ to
assemble $K$.  For diagonal or low-rank covariances the per-pair cost
reduces to $\mathcal O(d)$ or $\mathcal O(k d)$, giving
$\mathcal O(N^2 d)$ kernel assembly and again $\mathcal O(N^3)$ training.

\paragraph{Full Wasserstein GP.}
Wasserstein GPs build kernels from optimal-transport distances between
empirical distributions
\citep{mallasto2017learning,bachoc2017gaussian,Panaretos2019,candelieri2022gaussian}.
Computing the pairwise cost matrix between $U_i$ and $U_j$ is
$\mathcal O(M^2 d)$, and solving the OT problem (e.g. with Sinkhorn or
network simplex) requires at least $\mathcal O(M^2)$ per iteration and
$\mathcal O(M^2)$ memory \citep{peyre2019computational}.  In practice this yields a
per-pair distance cost of order $\mathcal O(M^2 d)$ to
$\mathcal O(M^3)$, and thus a total cost of
\[
\mathcal O(N^2 M^2 d) \quad \text{to} \quad \mathcal O(N^2 M^3)
\]
for computing $K$, plus the usual $\mathcal O(N^3)$ for GP training.
This quadratic (or worse) dependence on $M$ motivates the separable and
sliced variants below.

\paragraph{\PWAGP (separable Wasserstein kernel).}
Our \PWAGP uses a separable kernel built from one-dimensional
Wasserstein distances along each coordinate, still rooted in the OT
geometry of \citet{Panaretos2019}.  In 1D, the $W_2$ distance between
empirical distributions of size $M$ reduces to sorting, which costs
$\mathcal O(M \log M)$ per marginal.  For $d$ dimensions, a naive
implementation therefore costs $\mathcal O(d M \log M)$ per pair and
$\mathcal O(N^2 d M \log M)$ in total to assemble $K$, which is much
cheaper than full OT as soon as $M$ is moderately large.  Training again
costs $\mathcal O(N^3)$.

\paragraph{\PCPWA (PCA-based separable Wasserstein kernel).}
\PCPWA first computes a low-rank PCA basis for the pooled samples and
then applies the separable 1D Wasserstein construction along the leading
$k$ principal directions.  The up-front PCA on all samples costs
$\mathcal O(N M d^2)$ (or less with randomized SVD), but this is a
one-time cost.  Thereafter, the per-pair distance uses only $k$
one-dimensional marginals, for a cost $\mathcal O(k M \log M)$ per pair
and $\mathcal O(N^2 k M \log M)$ in total.  Since typically $k \ll d$,
\PCPWA can be substantially cheaper than \PWAGP in high dimensions while
also adapting to the intrinsic low-dimensional structure of the clouds.

\paragraph{Sliced Wasserstein GP.}
The sliced Wasserstein GP \citep{meunier2022slicedwasserstein} draws
$R$ random projection directions
$\{\theta_r\}_{r=1}^R \subset \mathbb S^{d-1}$ and approximates the
Wasserstein distance by averaging 1D Wasserstein distances along these
projections.  Projecting all samples in $U_i$ onto $\theta_r$ costs
$\mathcal O(M d)$ per direction and $\mathcal O(R M d)$ in total per
cloud.  These projections can be precomputed and reused across pairs.
Given the projected samples, computing 1D distances between $U_i$ and
$U_j$ across all $R$ directions costs $\mathcal O(R M \log M)$ per pair.
The overall kernel assembly cost is therefore
\[
\mathcal O(N R M d) \quad \text{(precomputation)} \;+\;
\mathcal O(N^2 R M \log M) \quad \text{(pairwise distances)},
\]
followed by $\mathcal O(N^3)$ training.  In our code $R$ is treated as a
moderate constant, so the dominant dependence is
$\mathcal O(N^2 M \log M)$.

\paragraph{KME and MMD GPs.}
The KME and MMD kernels embed each empirical distribution into an RKHS
using sample averages of feature maps
\citep{muandet2017kme,szabo2016distreg}.  With a Gaussian base kernel
and empirical distributions of size $M$, the kernel between $U_i$ and
$U_j$ typically uses a double sum over samples,
\[
k_{\mathrm{KME}}(U_i,U_j)
  \;\approx\; \frac{1}{M^2} \sum_{a=1}^{M} \sum_{b=1}^{M}
    \exp\!\Bigl(-\frac{\|u_{ia} - u_{jb}\|^2}{2\ell^2}\Bigr),
\]
and similarly for the squared MMD.  Each kernel evaluation is therefore
$\mathcal O(M^2 d)$, and assembling $K$ costs
$\mathcal O(N^2 M^2 d)$, again followed by $\mathcal O(N^3)$ training.
Random features can reduce this to $\mathcal O(N^2 D d)$ with $D$
features, but in our experiments we use the exact empirical embedding.

\paragraph{Low-rank Fréchet EIV regression.}
The low-rank Fréchet errors-in-variables (EIV) regression of
\citet{song2023frechet} represents each distributional covariate via a
rank-$r$ approximation in a suitable Hilbert space.  Constructing these
low-rank features from $M$ samples per distribution costs roughly
$\mathcal O(N M r)$, after which Fréchet regression (or GP regression on
the $r$-dimensional coefficients) operates on $N$ inputs of dimension
$r$.  Kernel assembly is therefore $\mathcal O(N^2 r)$, comparable to a
regular GP in $r$ dimensions, with the usual $\mathcal O(N^3)$ training
cost.  When $r \ll d$ and $M$ is moderate, this yields a competitive
complexity to our PCA-based \PCPWA construction.
 
Table~\ref{tab:complexity} summarizes the leading-order complexity of
the main models as a function of the number of clouds $N$, samples per
cloud $M$, and dimension $d$ (ignoring the shared $\mathcal O(N^3)$ GP
training cost).  Regular and uncertain-input GPs
\citep{williams2006gpml,girard2002gaussian,mchutchon2011gaussian}
are cheapest in terms of $M$, since they compress each distribution to a
finite-dimensional summary.  Aggregated GP is prohibitively expensive as
soon as $M$ grows.  Among fully distributional methods, our separable
Wasserstein kernels (\PWAGP, \PCPWA), together with \SWGP
\citep{meunier2022slicedwasserstein}, enjoy \emph{linear} or near-linear
dependence on $M$ (up to logarithmic factors) thanks to one-dimensional
Wasserstein computations
\citep{Panaretos2019,peyre2019computational}.  In contrast, full \WGP, KME, and MMD
kernels \citep{mallasto2017learning,bachoc2017gaussian,muandet2017kme,szabo2016distreg}
scale at least quadratically in $M$.  This helps explain the empirical
observation in our simulated and accelerator experiments that \PWAGP
and \PCPWA provide competitive errors-in-variables performance and
well-calibrated uncertainty while remaining computationally tractable.

\begin{table}[t]
\centering
\label{tab:complexity}
\resizebox{\columnwidth}{!}{%
\begin{tabular}{lcc}
\toprule
Method & Effective \# inputs & Kernel assembly cost (big-O) \\
\midrule
Regular GP on means \citep{williams2006gpml}
    & $N$     & $\mathcal O(N^2 d)$ \\
Aggregated GP on samples \citep{williams2006gpml}
    & $N M$   & $\mathcal O(N^2 M^2 d)$ (matrix) $+$ $\mathcal O(N^3 M^3)$ (training) \\
Uncertain-input GP \citep{girard2002gaussian,mchutchon2011gaussian}
    & $N$     & $\mathcal O(N^2 d^3)$ (full cov.) or $\mathcal O(N^2 d)$ (diag) \\
Full \WGP \citep{mallasto2017learning,bachoc2017gaussian,candelieri2022gaussian}
    & $N$     & $\mathcal O(N^2 M^2 d)$ to $\mathcal O(N^2 M^3)$ \\
\PWAGP (separable Wasserstein)
    & $N$     & $\mathcal O(N^2 d M \log M)$ \\
\PCPWA (PCA-based PWA)
    & $N$     & $\mathcal O(N M d^2)$ (PCA) $+$ $\mathcal O(N^2 k M \log M)$ \\
\SWGP (sliced Wasserstein) \citep{meunier2022slicedwasserstein}
    & $N$     & $\mathcal O(N R M d) + \mathcal O(N^2 R M \log M)$ \\
KME-GP \citep{muandet2017kme}
    & $N$     & $\mathcal O(N^2 M^2 d)$ \\
MMD-GP \citep{szabo2016distreg}
    & $N$     & $\mathcal O(N^2 M^2 d)$ \\
Low-rank Fréchet EIV regression \citep{song2023frechet}
    & $N$     & $\mathcal O(N M r) + \mathcal O(N^2 r)$ \\
\bottomrule
\end{tabular}%
}
\caption{Leading-order cost of assembling the $N\times N$ kernel matrix
for the different methods, assuming $m_i \asymp M$ samples per cloud.
All methods then incur a shared $\mathcal O(N^3)$ cost for exact GP
training.}
\end{table}

\end{document}